\documentclass[USenglish]{article}

\usepackage[utf8]{inputenc}
\usepackage{mathrsfs}
\usepackage{dsfont}
\usepackage{amsthm}
\usepackage{amsmath}
\usepackage{algorithm}
\usepackage{algpseudocode}
\usepackage{stmaryrd}
\usepackage{graphicx}
\usepackage{mathtools}
\usepackage{caption}
\usepackage{comment}
\usepackage{geometry}
\newcommand{\conf}{ \mathscr{C} }
\newcommand{\state}{ \ensuremath{\mathscr{S}} }
\newcommand{\sep}{ \mathds{APART} }
\newcommand{\R}{ \mathds{R} }
\newcommand{\Q}{ \mathds{Q} }
\newcommand{\prob}{ \mathds{P} }
\newcommand{\N}{ \mathds{N} }
\newcommand{\Rs}{ \mathds{R}^2 }
\newcommand{\oi}{ \{ 0;1 \} }
\newcommand{\angles}{ [ 0 ; \pi [}
\newcommand{\alg}{\ensuremath{\mathscr{A}}}

\newtheorem{lemma}{Lemma}
\newtheorem{definition}{Definition}

\newtheorem{theorem}{Theorem}
\newtheorem{conjecture}{Conjecture}

\begin{document}

\huge
 \begin{center}
Asynchronous Scattering

\hspace{10pt}

\Large
Ulysse L\'echine$^{1,2}$, S\'{e}bastien Tixeuil$^2$ \\

\hspace{10pt}

\small  
$^1$) \'{E}cole Normale Sup\'{e}rieure de Lyon, Lyon, France\\
ulysse.lechine [at] ens-lyon.fr\\
$^2$) Sorbonne Universit\'{e}, CNRS, LIP6, FR-75005, Paris, France\\
firstname.lastname@lip6.fr

\end{center}

\hspace{10pt}

\normalsize



\begin{abstract}
In this paper, we consider the problem of scattering a swarm of mobile oblivious robots in a continuous space. We consider the fully asynchronous setting where robots may base their computation on past observations, or may be observed by other robots while moving.

It turns out that asynchronous scattering is solvable in the most general case when both vision (the ability to see others robots positions) and weak local multiplicity detection are available. 
In the case of a bidimensional Euclidean space, ASYNC scattering is also solvable with blind robots if moves are rigid. 
Our approach is constructive and modular, as we present a proof technique for probabilistic robot protocols that is of independent interest and can be reused for other purposes. 

On the negative side, we show that when robots are both blind and have no multiplicity detection, the problem is unsolvable, and when only one of those is available, the problem remains unsolvable on the line.
\end{abstract}

\clearpage

\section{Introduction}

We consider mobile oblivious robots that evolve in a bidimensional Euclidean space~\cite{SuzukiY99,FlocchiniPS19}. Those robots are modeled as dimensionless points and repeatedly execute cycles that consist of Look-Compute-Move phases: they first obtain a snapshot of their environment using visual sensors (Look), then compute their next tentative location (Compute), and finally move toward their computed destination (Move). After a Look-Compute-Move cycle is complete, the local memory of a robot $r$ is erased before $r$ starts a new cycle (the robot is said \emph{oblivious}). 
In the semi-synchronous model (SSYNC in the current terminology~\cite{SuzukiY99,FlocchiniPS19}), a non-empty subset of robots is chosen by the scheduler (considered as an adversary) to simultaneously and atomically execute a Look-Compute-Move cycle. An even more restricted model is the fully synchronous model (FSYNC~\cite{SuzukiY99,FlocchiniPS19}), where the subset of scheduled robots is the full set of robots. While SSYNC and FSYNC model make it easier to write and prove robot protocols, they do not represent well the asynchronous behavior of actual robotic entities. For that purpose, the asynchronous model (ASYNC~\cite{FlocchiniPSW05,FlocchiniPS19}) has been introduced: in that model, each phase of the Look-Compute-Move cycle of a robot may occur at any time. In particular, the following events may occur: When a robot executes a Look phase, the other robots are not necessarily still; also, a robot follows a straight path to reach its destination, but not necessarily at a constant speed (though it cannot move back).  

The visual sensors of a robot $r$ output information about the $r$'s environment in an ego-centered coordinate system (which we denote in the sequel as \emph{frame}). It is generally assumed that robots' coordinate systems do \emph{not} share a common direction (same "North-South axis"), a common orientation (same "North"), a common chirality (same handedness), or even a common unit distance. 
In this ego-centered coordinate system, the output of the visual sensors consists in (possibly partial) information about the positions of the other robots (the current robot always being at the origin of the coordinate system). A first criterium to classify vision sensors is its \emph{range}~\cite{FlocchiniPSW05}: how far can it see other robots positions? In this paper we consider two kinds of robots: \emph{perceptive} robots can see every location occupied by at least one robot, while \emph{blind} robots cannot see any other location than the origin of the coordinate system. A second criterium for vision sensors is its ability to detect multiple robots placed at the exact same location~\cite{IzumiIKO13}: \emph{strong} multiplicity detection returns the exact number of robots placed on a given location, \emph{weak} multiplicity detection returns \emph{true} if a location hosts at least two robots, and \emph{false} if it contains exactly one robot, while \emph{no} multiplicity detection returns no information about a position other than whether it contains at least one robot. Multiplicity detection is \emph{local} if it can only be used on the observing robot's location, and is \emph{global} if it can be used on every occupied location within the range of the vision sensor (and this range is greater than $0$).

During the move phase, the move of the robot is either \emph{rigid} or \emph{flexible}~\cite{FlocchiniSVY16}. If robot moves are considered rigid, a robot always reaches its computed destination (possibly after some unbounded time in the ASYNC model). If robot moves are flexible, then the scheduler can interrupt a robot $r$ before it reaches its destination. Then $r$ starts a new Look-Compute-Move cycle. In the classical model~\cite{SuzukiY99}, the scheduler may not stop $r$ before it moves at least some distance $\delta>0$ (unknown to $r$) unless $r$ target destination is closer than $\delta$ (in that case, $r$ reaches its destination). In this paper, we instead consider a more extreme case where robot moves can be \emph{infinitesimal}: each robot $r$ can be interrupted after it moves some distance $\epsilon>0$, for any such $\epsilon$. 

The problem we consider is that of \emph{scattering}: from any initial situation where robots may occupy the same locations, any execution of a scattering protocol ends up in a situation where all robots occupy distinct positions. Termination (that is, reaching a fixed global state where all robots remain still) is not required, as long as robots remain at distinct locations. 
A deterministic scattering protocol for oblivious mobile robots is feasible if and only if there does not exists \emph{clone} robots~\cite{SuzukiY99}: two robots are clones if they start at the same position, are given the same coordinate system, and are always activated simultaneously by the scheduler. Two deterministic clone robots never separate. By contrast, a robot with no clone can simply move "North" by one distance unit when activated by the scheduler to separate from another robot initially located at the same position.
As clone robots cannot be avoided in general, most subsequent work studied probabilistic protocols. 
The first probabilistic scattering for the SSYNC model is due to Dieudonn\'{e} and Petit~\cite{DieudonneP09}: anytime a robot is activated, it tosses a coin, if the coin is $0$, the robot stays in place, otherwise it moves arbitrarily within its Vorono\"{\i} cell. Then, Cl\'{e}ment \emph{et al.}~\cite{ClementDPIM10} proposed an expected constant time scattering protocol for the SSYNC model: assuming knowledge of $n$, the number of robots, anytime a robot is activated, it chooses a destination randomly among $n^2$ ones, still within its Vorono\"{\i} cell. Recently, Bramas and Tixeuil~\cite{BramasT17} considered the random bit complexity of mobile robot scattering in the SSYNC model: anytime a robot is activated, it chooses randomly among $f$ destinations ($f$ depends on the observed configuration, so no knowledge of $n$ is assumed) within its Vorono\"{\i} cell. Bramas and Tixeuil~\cite{BramasT17} show that their solution is the best possible in SSYNC (both in terms of coin tossing and round complexity) when $n$ is not known.  
Overall, it turns out that the scattering problem for oblivious mobile robots was only investigated in the SSYNC model. The reason is that for all aforementioned protocols, the core proof argument relies on the fact that once two robots separate, they never occupy the same position ever again, hence the number of robots on a given multiplicity point \emph{always} decreases. In the ASYNC model, this property is no longer guaranteed using previous approaches. 
For example, consider the protocol of Dieudonn\'{e} and Petit~\cite{DieudonneP09}, and two stacks of $2$ and $3$ robots, respectively. First the stack of two robots is activated: one robot chooses to move to destination $d$ and the other chooses to stay, \emph{but the moving robot, called $r_1$, is put on hold}. Then the stack of $3$ robots is activated, $2$ of those $3$ robots, say $r_2$ and $r_3$ decide to move, and the other one decides to stay. Now $r_2$ is activated again, but its observation of other positions is updated (and so is its Vorono\"{\i} cell). So, it may turn out that $r_2$ chooses $d$ as its destination. Now, both $r_1$ and $r_2$ move to $d$, creating a \emph{new} multiplicity point.
So, solving the problem in ASYNC requires a fresh approach, at least with respect to the proof techniques.

We investigate the possibility of asynchronous scattering in various settings. In particular, we consider both perceptive and blind robots, both weak local multiplicity detection and no multiplicity detection, and both rigid and flexible moves. It turns out that ASYNC scattering is solvable in the most general case when both vision and weak local multiplicity detection are available. In the case of a bidimensional Euclidean space, ASYNC scattering is also solvable with blind robots if moves are rigid. Our approach is constructive and modular, as we present a proof technique for probabilistic robot protocols that is of independent interest and can be reused for other purposes. On the negative side, we show that when robots are both blind and have no multiplicity detection, the problem is unsolvable, and when only one of those is available, the problem remains unsolvable on the line. Unlike previous approaches~\cite{DieudonneP09,ClementDPIM10,BramasT17}, we consider the fully realistic ASYNC model, and unlike Clement \emph{et al.}~\cite{ClementDPIM10}, we do not consider that robots are aware of $n$. Our positive result do not assume that robots share the same chirality, while impossibility results remain valid if robots share a common chirality. 

The remaining of the paper is structured as follows. Section~\ref{sec:framework} details the mathematical framework we work with. Section~\ref{sec:proof} features a general result, which we believe is of independent interest, that is later used to prove our results about scattering in Section~\ref{sec:scattering} and \ref{sec:flexible}. Section~\ref{sec:complexity} presents preliminary complexity results, while Section~\ref{sec:computability} takes a computability perspective. Finally, Section~\ref{sec:conclusion} provides some concluding remarks.

\section{Framework}
\label{sec:framework}

\subsection{Notations}

$\R$ denotes the set of real numbers, $\Q$ the set of rational numbers, and $\N$ the set of natural numbers. The set of natural numbers except 0 is denoted by $\N^\star$. 
Let f be a function, we can also denote it as $f(\bullet)$. $\| \bullet \| : \Rs \hookrightarrow \R$ denotes the standard euclidean norm over $\Rs$, and $\langle \bullet , \bullet \rangle$ denotes the standard scalar product over $\Rs \times \Rs$. Let $E$ be a set, then $E^*$ denotes the set of all finite sequences of elements of $E$, $E^{\N}$ denotes the set of all sequences over $E$, and the set of infinite sequences $E^{ \infty} = E^{\N} \cup E^*$. Let $e \in E_1 \times E_2$ let $p$ and $q$ be the two elements such that $e=(p,q)$, then $e(1)= p$ and $e(2) = q$, we generalize this notation to the case where $e$ is composed of more than two elements.

\subsection{Execution Model}

We consider the oblivious Look-Compute-Move model (\emph{a.k.a.} the $\mathcal{OBLOTS}$ model) in an asynchronous setting. 

The current state of a robot is modelled as a 6-tuple $(l,d,r,s,a,c) \in \state$ such that:
\begin{enumerate}
\item $l\in \Rs$ is the current \emph{location} of the robot,
\item $d\in \Rs$ is the \emph{destination} of the robot,
\item $r\in \oi$ is $1$ if the robot is activated and $0$ otherwise,
\item $s\in \R, s>0$ is the \emph{scale} of the referential of the robot (the size of the unit distance),
\item $a\in \angles$ is the \emph{angle} of the referential of the robot (where is the north),
\item $c\in \oi$ is $1$ if the robot is right-handed and $0$ otherwise.
\end{enumerate}

We define the set of configurations $\conf$ as the set of combinations of the local states of the robots where at most one robot is activated in each configuration: $\conf = \{ (s_1, \ldots ,s_n) \in \state ^ n | \exists i \in \llbracket 1 ; n \rrbracket (s_i(3) = 1 \Rightarrow \forall j \in \llbracket 1 ; n \rrbracket , j \neq i \Rightarrow s_j(3) = 0 )  \}$. We say that the $k$-th robot is \emph{activated} if and only if $C(k)(3) = 1$, otherwise it's \emph{inactivated}. Having a single robot activated at any time does not weaken the ASYNC model as it still allows to simulate simultaneous executions, yet it permits to simplify notations in the sequel. 

Let $\sep$ be the set of configurations such that there is no location occupied by two or more robots, and every robot is immobile (its destination and its location are equal): $\sep = \{ (s_1, \ldots , s_n) \in \conf$ | $\forall i,j \in \llbracket 0 ; 1 \rrbracket ^ 2 , s_i(1)=s_i(2) \wedge s_i(1) \neq s_j(1)\}$. 

Let $\alg$ be the algorithm all robots execute. We consider two settings for visual sensors:
\begin{itemize}

\item If \emph{vision} is assumed $\alg$ receives as inputs the set of points where there is at least one robot. Said differently given a point in $\Rs$ robots know if there is at least one robot there but not how many. If \emph{vision} is not assumed the robot is said to be \emph{blind}
\item If \emph{local weak multiplicity} is assumed, $\alg$ receives as input a bit indicating if there is another robot on the same position as the robot running $\alg$.
\end{itemize}

An \textbf{execution} is an element of $\conf ^ \infty = C_1 \rightarrow C_2 \rightarrow C_3 \ldots $ such that every transition $C_i \rightarrow C_{i+1}$ is valid. Now, $C \rightarrow D$ is a valid transition for $\alg$ if and only if there exists a unique $k\in \N$ such that $C(k) \neq D(k)$ (\emph{i.e.} only one robot did something) and all of the following properties are satisfied:
\begin {itemize} 
\item At least one of the following conditions is true:
\begin{itemize}
\item \textbf{choice of frame} $C(k) = (p,p,0,r,\theta,b)$ and $D(k)=(p,p,1,r',\theta',b')$: the scheduler assigns a frame to robot $k$,
\item \textbf{computation} $C(k) = (p,p,1,r,\theta,b)$ and $D(k) = (p,q,0,r,\theta,b)$: the robot $k$ computes its destination $q$ using \alg, 
\item \textbf{ready to move} $C(k) = (p,q,0,r, \theta,b)$ where $p \neq q$ and $D(k) = (p,q,1,r, \theta,b)$: robot $k$ is about to move,
\item \textbf{move} $C(k) = (p,q,1,r, \theta)$ where $p \neq q$ and $D(k) = (q',q,0,r, \theta)$ where $q' \in ]p;q]$: robot $k$ moves towards its destination,
\item \textbf{move and interrupt} $C(k) = (p,q,1,r, \theta)$ where $p \neq q$ and $D(k) = (q',q',0,r, \theta)$ where $q' \in ]p;q]$: robot $k$ moves towards its destination, but may be interrupted prematurely by the scheduler.
\end{itemize}
\item If more than one condition is true, then computation, move, and move and interrupt transitions have priority over the other two.
\end{itemize}

Notice that transitions are of two types, the ones that inactivate robot $k$ (computation, move, move and interrupt), and the ones that activate it (choice of frame, ready to move). Due to the higher priority of the former set of transitions, the property that at most one robot is activated in any configuration is preserved throughout the execution. 
Starting from configuration $C_i$, if $\alg$ is deterministic, then there is at most one reachable configuration. If $\alg$ is a probabilistic algorithm, then there can be more than one reachable configurations. 
An execution $E$ is \textbf{fair} if all robots perform the \emph{computation} transition infinitely many times.  Given an algorithm $\alg$, we denote by $\mathscr{E}_{\alg}$ the set of executions conforming to $\alg$.\\

Let us now define \textbf{schedulers}. Intuitively the scheduler chooses when robots are activated, by how much they move when the movement is not rigid, etc. Let $\alg$ be an algorithm, let $\sigma : \conf ^ * \hookrightarrow \conf$ be a function. If $\forall C_1, \ldots, C_k, C_k \rightarrow \sigma(C_1, \ldots , C_k)$ is a valid transition for $\alg$, then $\sigma$ is said to be a scheduler for $\alg$. This ensures that all executions conform to $\alg$.

Given an initial configuration $C$, we define execution $E_{\sigma}^C$ as the execution starting from $C$ according to scheduler $\sigma$, and $E_{\sigma}$ is the set of all possible executions. If $\alg$ is deterministic, then $E_{\sigma}^C(0)=C$, and $\forall i \in \N^\star, \quad E_{\sigma}^C(i) = \sigma(E_{\sigma}^C(0), \ldots , E_{\sigma}^C(i-1))$. If $\alg$ is probabilistic, then $E_{\sigma}^C(i)$ is a random variable defined the same way as in the deterministic case, and $\sigma$ is a probabilistic function, that is to say it returns its output according to a certain distribution law given by the algorithm $\alg$ for the transitions of the kind \emph{computation}.

In the deterministic setting, $\sigma$ is a \textbf{fair scheduler} if $\sigma$ is a scheduler and $E_{\sigma}$ is fair. In the probabilistic setting, $\sigma$ is fair if $\prob (E_{\sigma} \text{ is fair})=1$. 
Given Algorithm $\alg$, the set of all possible schedulers conforming to $\alg$ is noted $\Sigma_{\alg}$. If the algorithm is implied, we simply write $\Sigma$. If a scheduler never chooses a \emph{move and interrupt} transition, the movements are called \textbf{rigid}.

The fine grained approach we used for transitions permits to simulate an arbitrary number of robots running simultaneously (even when other are moving). Two concurrent robots $A$ and $B$ could first sequentially execute \emph{choice of frame} and \emph{computation}, and at a later time execute (again sequentially) \emph{ready to move} and \emph{move}. A robot $A$ can be seen while moving by robot $B$ interleaving \emph{move} by $A$, \emph{computation} by $B$, and \emph{ready to move} by $A$. However, our modeling with at most one robot activated in any configuration permits to simplify proof notations.

\section{Theorem of Choice}
\label{sec:proof}

\begin{definition}[Safe Zone]
Let $\alg$ be an algorithm, let $F \subset \conf$, $F$ is called a \emph{safe zone} for $\alg$ if valid transistions only lead to configurations of $F$. Formally: $\quad \forall (\sigma , C) \in \Sigma_{\alg} \times F, \forall i \in \N, E^{C}_{\sigma}(i) \in F \Rightarrow E^{C}_{\sigma}(i+1) \in F$.
\end{definition}

\begin{lemma}
\label{safe}
Let $F \subset \conf $ and $\alg$ be an algorithm. If there exists a positive probability $p$ such that, starting from any configuration and for any scheduler, every execution reaches a configuration in $F$, then configurations in $F$ are reached with probability $1$. Formally: $\exists p \in ]0;1], \forall (C, \sigma) \in \conf \times \Sigma_{\alg}, \prob (\exists i \in \N, E^{C}_{\sigma}(i) \in F ) \geq p \Rightarrow \forall (C, \sigma ) \in \conf  \times \Sigma_{\alg}, \prob (\exists i \in \N, E^{C}_{\sigma}(i) \in F ) = 1$.
\end{lemma}

\begin{proof}
Let $p$ be a positive real number. For all $(C,\sigma) \in \conf \times \Sigma$, we have that:

\begin{align*}
\prob ( \exists i \in \N, E^{C}_{\sigma}(i) \in F ) = \ldots \\ 
& \hspace{-4cm} = \prob \left( \bigcup_{ i \in \N}  \text{from C, following $\sigma $, F is reached \emph{exactly} at the $i$-th transition} \right) \\
 & \hspace{-4cm} = \sum_{ i \in \N} \prob \left( \text{from C, following $\sigma $, F is reached \emph{exactly} at the $i$-th transition}\right) \geq p
\end{align*}

This implies that for all $(C,\sigma) \in \conf \times \Sigma$:
\begin{align*}
& \exists N \in \N, \sum_{i=0}^{N} \prob \left(\text{from $C$, following $\sigma$, $F$ is reached \emph{exactly} at the $i$-th transition} \right) \geq \frac{p}{2} \\
& \Rightarrow \exists N \in \N, \prob \left( \text{from $C$, following $\sigma$, $F$ is reached before  $N$-th transition}\right) \geq \frac{p}{2} 
\end{align*}

For each configuration $C$ and scheduler $\sigma$, we can thus define $N(C,\sigma)$ as the smallest integer $N$ such that: $\prob ( \text{from $C$, following $\sigma$, $F$ is reached before the $N$-th transition}) \geq \frac{p}{2}$. 

We now prove the lemma. Let $C$ be a configuration and $\sigma$ a scheduler. With execution $E^{C}_{\sigma}$ and $i \in \N$,
we define $shift (\sigma,E^{C}_{\sigma},i)$ as $\sigma(E^{C}_{\sigma}(0), \ldots , E^{C}_{\sigma}(i), \bullet)$. To simplify notations, $C_i$ denotes $E^{C}_{\sigma}(i)$ in the sequel.

Let $E^{C}_{\sigma}$ be a random execution starting from $C$. We define $(n_i)_{i\in \N}$ and $(\sigma_i)_{i \in \N}$ as $n_0=0$, $\sigma_0=\sigma$, and $\forall i \in \N^\star, n_{i+1}= n_{i} + N(C_{n_i}, \sigma_i)$, and $\sigma_{i+1}= shift (\sigma,E^{C}_{\sigma},n_{i+1})$. 
Now, from $C \in \conf$, following $\sigma$, we run $n_1$ transitions so that $F$ is reached with probability $\frac{p}{2}$. If we did reach $F$ we stop. Otherwise, from $C_{n_i}$, following $\sigma$, we run $n_2$ transitions so that $F$ is reached with probability $\frac{p}{2}$. If we did reach $F$, we stop. Otherwise we repeat the process. Overall the probability that $F$ is never reached is:

\begin{align*}
\prob(\text{from $C$, following $\sigma $, we never reach $F$})& = \prob(\forall i, E^{C}_{\sigma}(i) \notin F )\\
& = \prob(\bigcap_{i \in \N} \forall j \in \llbracket n_i, n_{i+1} \llbracket, E^{C}_{\sigma}(j) \notin F )\\
& = \prob(\bigcap_{i \in \N} \forall j \in \llbracket n_i, n_{i+1} \llbracket, E^{C_{n_i}}_{\sigma_i}(j) \notin F )\\
&\leq \prod_{i \in \N} 1-\frac{p}{2}\\
& = 0
\end{align*}
\end{proof}

The following theorem is a direct corollary of Lemma~\ref{safe} using the closure property of a safe zone.
\begin{theorem}[Choice]
\label{thm:choice}
Let $\alg$ be an algorithm, let $F \subset \conf $ be a safe zone for $\alg$, if $\exists p \in ]0;1], \forall (C, \sigma ) \in \conf  \times \Sigma_{\alg}, \prob ( \exists i \in \N, E^{C}_{\sigma}(i) \in F ) \geq p $, then $\forall (C, \sigma ) \in \conf \times \Sigma_{\alg}, \prob ( \exists N \in \N, \forall i \in \N, i \geq N \Rightarrow E^{C}_{\sigma}(i) \in F ) = 1$
\end{theorem}

So, if $F$ be a safe zone, if there is a positive probability $p$, such that for any scheduler $\sigma$ and any \textbf{reachable} configuration $C$, any execution from $C$ that follows $\sigma$ reaches a configuration in $F$ with probability $p$, then $F$ is reached with probability 1. A reachable configuration is a configuration with a positive probability of occurring in the execution.
Note that the converse is also true: If $F$ is reached with probability $1$, then starting from any reachable configuration and following any scheduler, $F$ is reached with probability $1>0$.

While Theorem~\ref{thm:choice} is very general, we mostly use it for the purpose of scattering robots: our safe zone $F$ is a set of configurations where robots remain separated whatever the subsequent execution, and we want that starting from any initial configuration $C$ (possibly with already moving robots) and following any scheduler $\sigma$, algorithm $\alg$ reaches $F$ with probability $1$. Then, using Theorem~\ref{thm:choice}, it becomes sufficient to exhibit a pattern of random bits used by $\alg$ so that the robots reach $F$ with probability $p>0$, and prove that $p$ is independent of $C$ and $\sigma$.

\clearpage

\section{Blind Rigid ASYNC Scattering in $\R^2$}
\label{sec:scattering}

In this section, we consider the rigid setting when all robots are blind but are endowed with weak local multiplicity detection.

\subsection{Algorithm}

\begin{algorithm}
\caption{Blind Rigid ASYNC Scattering in $\R^2$}\label{scattering}
\begin{algorithmic}[1]
\If {there is no other robot at my position} 
	\State Do nothing
\Else
	\State $i \gets 0$
	\State $r \gets rbit()$
	\While {$r=0$}
		\State $i \gets i+1$
		\State $r \gets rbit()$
	\EndWhile
	\State \textbf{move to \textbf{u}(i)}
\EndIf
\end{algorithmic}
\end{algorithm}

The pseudo-code for our algorithm is presented as Algorithm~\ref{scattering}. Lines $5$ and $8$ use the \textit{rbit()} function that returns a $1$ or a $0$ with equal probability. The while loop is used to select a random natural number such that value $n$ is obtained with probability $2^{-n}$. The keyword \textbf{move to} takes as input an element of $\Rs$ and updates the destination of the robot. Finally, $u(n)$ is a computable sequence of points in $\Q^2$ such that there exists a circle containing all $u(n)$, and $\forall i,j \quad u(i) \neq u(j)$ (such a sequence does exist).

\begin{figure}[h]
 \centering
 \includegraphics[width=.5\textwidth]{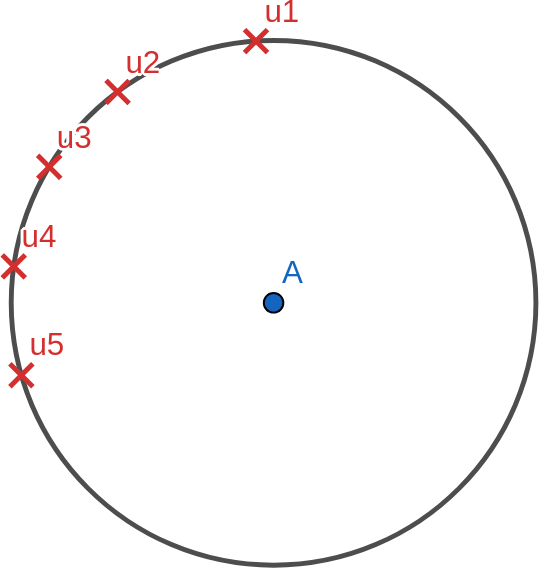}
 \caption{Five positions robot $A$ can go to}
 \label{cercle}
\end{figure}

\subsection{Proof of correctness}

\begin{lemma}
$\sep$ is a safe zone for Algorithm~\ref{scattering}.
\end{lemma}

\begin{proof}
Starting from a configuration in $\sep$, a robot does nothing when there is no other robot at its position (Line 2).
\end{proof}

\begin{lemma}
\label{sepalg}
With $F=\sep$, and $\alg=$Algorithm~\ref{scattering}, the conditions of Lemma~\ref{safe} are satisfied.
\end{lemma}

\begin{proof}
Let $\sigma, C \in \Sigma , \conf$. Step by step, we construct a sequence of positions that reach $\sep$ with probability $p>0$, independently of $\sigma$ and $C$. We denote by $(p_1, \ldots ,p_n)_t$ the positions occupied by the robots at step $t$, and $(q_1, \ldots , q_n)_t$ the positions the robots want to go to at step $t$. Initially, we define:
\begin{itemize}
\item $(p_1, \ldots ,p_n)_0 = (C(1)(1), \ldots, C(n)(1))$
\item $(q_1, \ldots , q_n)_0 = (C(1)(2), \ldots, C(n)(2))$
\end{itemize}
Suppose at step $t$, the current configuration is $D$, when the scheduler selects a \textbf{computation} transition from $D$ to $D'$ for a robot $k$ on a multiplicity point, the robot chooses its destination for step $t+1$, and we define:
\begin{itemize}
\item $(p_1, \ldots ,p_n)_{t+1} = (D(1)(1), \ldots, D(n)(1)), \left( = (D'(1)(1), \ldots, D'(n)(1)) \right)$
\item $(q_1, \ldots ,q_n)_{t+1} = (D'(1)(2), \ldots, D'(n)(2))$
\end{itemize}
There is a probability $\rho >0$ that $\forall i \in \llbracket 0 ; n \rrbracket, i\neq k \Rightarrow q \notin [p_i,q_i] \wedge q_i \notin [p,q]$. We call this event the $(\ast)$ event. This events means that the robot chose a destination that is not on another robot's path ($q \notin [p_i,q_i]$), and that does not make the robot cross another robot's destination ($q_i \notin [p,q]$). As there are at most $2n$ points that can be selected by $k$ and do not satisfy these two properties (see Figure~\ref{droit}), there is a probability at least $2^{-2n}$ of choosing a point that does satisfy those two conditions.

\begin{figure}[h!]
 \centering
 \includegraphics[width=.9\textwidth]{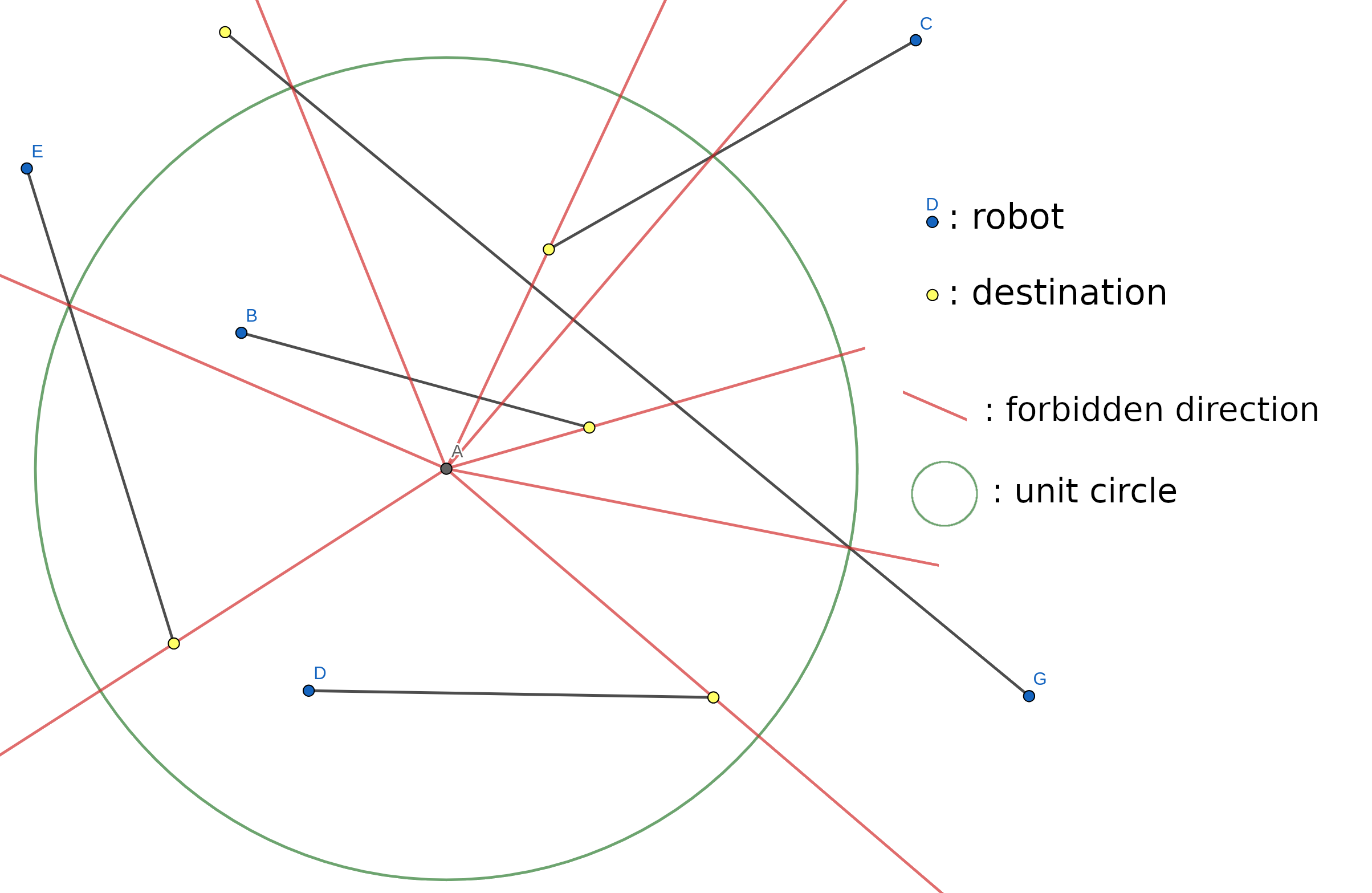}
 \caption{Possible destinations for $A$ so that ($\ast$) occurs.}
 \label{droit}
\end{figure}
 
Assume that at each step $(\ast)$ happens then we claim that we do at most $n$ steps. Indeed if we do $n+1$ steps then it means that a robot computed twice whilst not alone (from now on we only consider the computation the robot did "whilst not alone"), let us say it is the $k$-th robot who is the first robot to compute twice. If the robot computes twice then it means that at the time of its second computation it is not alone on its position, the h-th robot is also there. Let us call q the current position of the two robots (notice that q was also the result of the first computation of k), p the position where k was when it did its first computation, p' the position where h was when k did its first computation. Let $q'$ be the destination of h when k does its second computation. There are two possibilities : either the robot h had $q'$ in mind before the first computation of k or it had it after. 

\begin{itemize}
\item if $h$ had it before: then when $k$ did its first computation, it got $q \in [p';q']$, which is impossible by assumption (by the condition $q \notin [p_h,q_h]$),
\item if $h$ had it after: then $h$ did a computation while at position $e$ and chose $q'$ as a destination, with $q \in [e,q']$, which is impossible by assumption (by the condition $q_k \notin [e;q']$).
\end{itemize}

Since $(\ast)$ has probability at least $2^{-2n}$ of occurring, it has probability at least $2^{-2n^2}$ of occurring $n$ times in a row (notice that this probability is independent of $\sigma$ and $C$). Now, if it occurs $n$ times in a row, then $\sep$ is reached. Therefore $\sep$ has a probability at least $2^{-2n^2}$ of being reached from $C$, following $\sigma$.
\end{proof}

\begin{theorem}
\label{theo2Drigid}
Algorithm \ref{safe} solves the asynchronous rigid scattering problem in $\R^2$.
\end{theorem}

\begin{proof}
The proof is a direct application of theorem of choice to the result of Lemma \ref{sepalg}
\end{proof}

One remaining interesting question is a necessary and sufficient condition about $\{ u(i) \}_{i \in \N}$ such that Algorithm~\ref{safe} solves the ASYNC blind rigid scattering problem (see Figure \ref{conj}). One may think that the following property is true: "Algorithm~\ref{safe} solves ASYNC blind rigid scattering if and only if $\{ u(i) \}_{i \in \N}$ spans $\Rs$". While it is necessary (Algorithm~\ref{safe} does not solve scattering when all $u(i)$ are on the same line), it is not sufficient. Indeed, if $\forall n, n  = 0 [2] \Rightarrow u(0) = (1,0) \wedge n  = 1 [2] \Rightarrow u(1) = (0,1)$, Algorithm~\ref{safe} does not solve scattering (see Figure \ref{notsufficient}), yet $\{ u(i) \}_{i \in \N}$ is a base of $\Rs$. Instead, we propose the following conjecture: 

\begin{conjecture}
The $u(i)$ permit Algorithm~\ref{safe} to solve ASYNC blind rigid scattering if and only if there exists no finite set $C \subset \Rs$ of points such that from any points of $C$ there is a way for the scheduler to choose a frame that makes every $u(i)$ coincide with points of $C$.
\end{conjecture}

\begin{figure}[h]
    \centering
    \includegraphics[width=.3\textwidth]{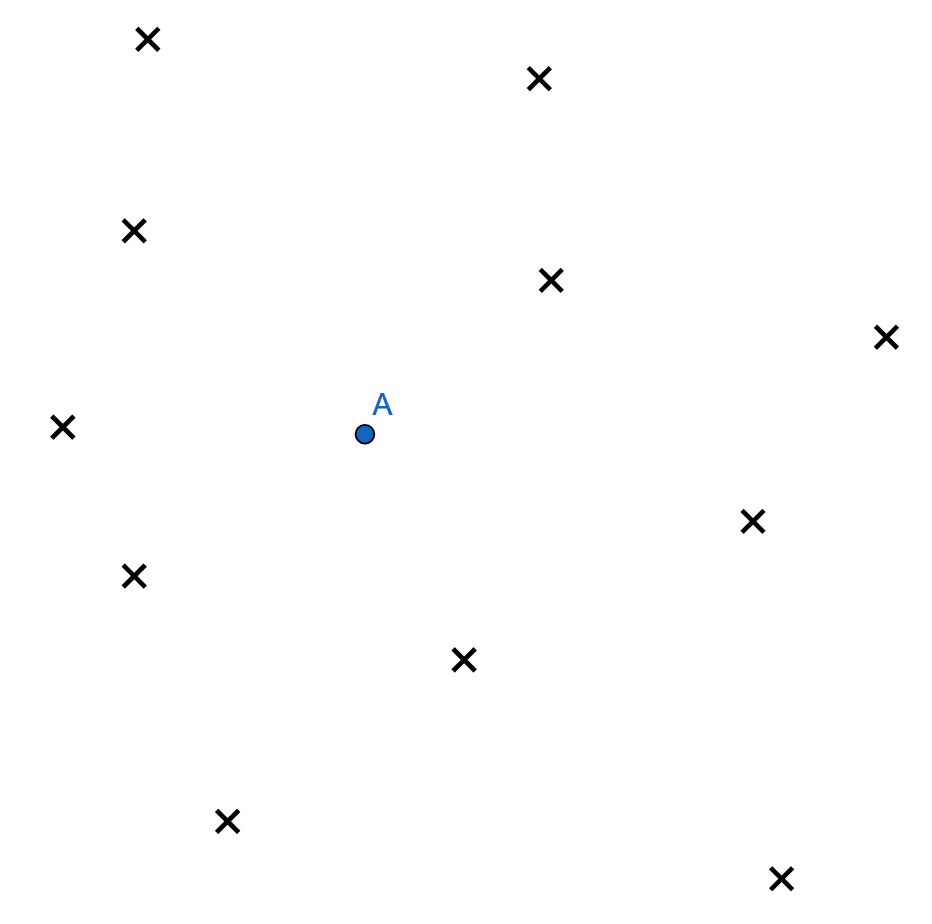}
    \caption{For which positions of the countably many crosses (that represent the possible destinations) does the algorithm solve ASYNC scattering?}
    \label{conj}
\end{figure}

\begin{figure}[h]
    \centering
    \includegraphics[width=.7\textwidth]{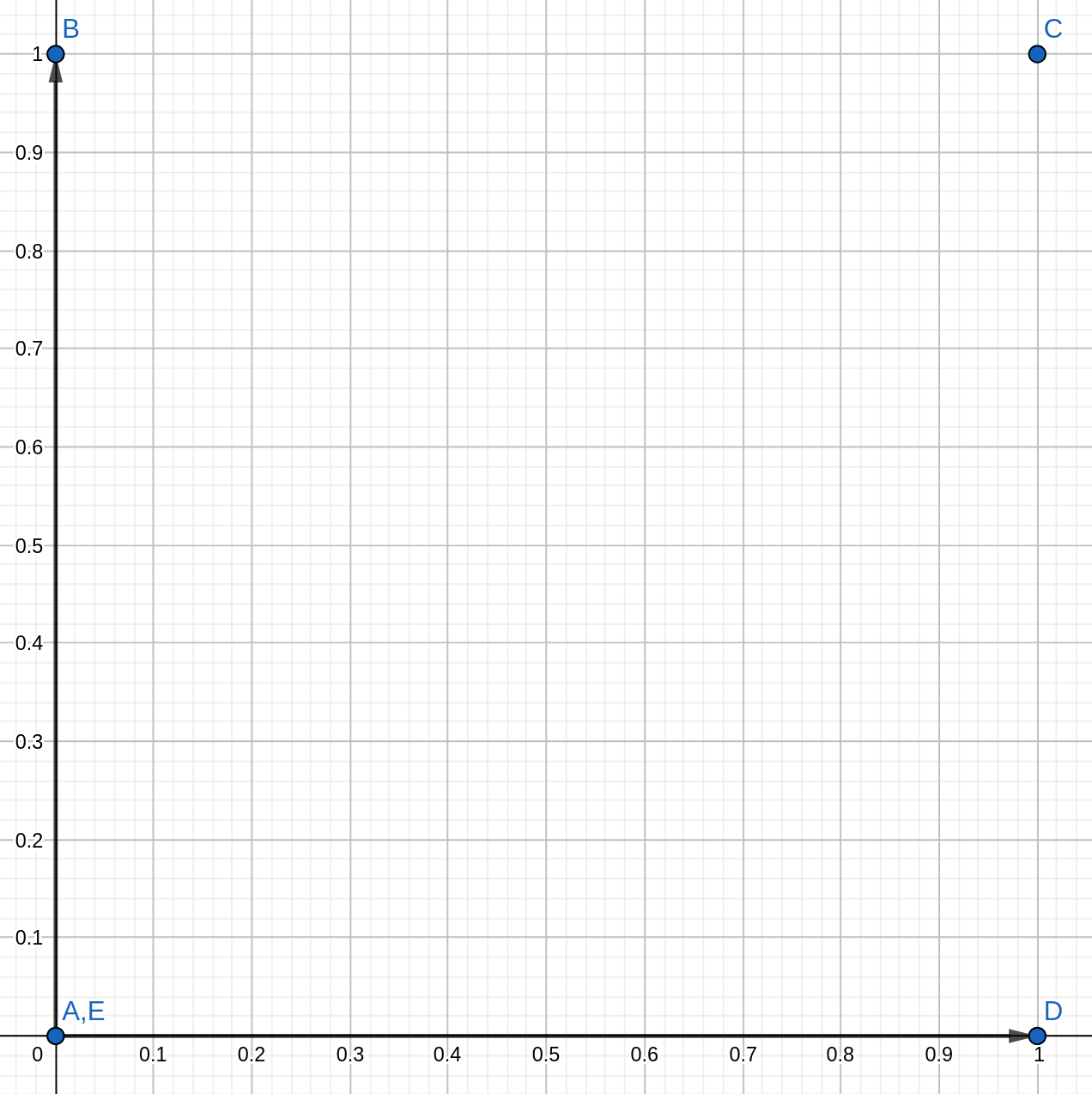}
    \caption{When $E$ moves, a symmetric situation occur, hence we have cyclic set of configurations.}
    \label{notsufficient}
\end{figure}

\clearpage

\section{Flexible Scattering in ASYNC}
\label{sec:flexible}

In this section, we consider the more difficult case of flexible (\emph{a.k.a.} non rigid) moves, where the scheduler may stop the robot before it reaches its destination (but not before the robot moves by some distance $\delta>0$). The previous blind solution (Algorithm~\ref{scattering}) may no longer be valid as the condition for not ending up in the same location now requires that robot \emph{trajectories} do not cross. So, the constant probability (from any configuration) of not reaching another robot's location is not guaranteed, and so neither is scattering. 
As a result, we devise new algorithms for this setting that are \emph{not} blind, first in the unidimensional case (Section~\ref{sec:unidimensional}), and then in the bidimensional case (Section~\ref{sec:bidimensional}).


\subsection{Flexible ASYNC scattering in $\R$}
\label{sec:unidimensional}

\begin{algorithm}
\caption{Flexible ASYNC Scattering in $\R$}\label{onescattering}
\begin{algorithmic}[1]
\If {there is no other robot at my position} 
	\State Do nothing
\Else
    \If{there is only one occupied location}
	    \State $r \gets rbit()$
		\State move by a unit distance to the right if $r=0$, to the left otherwise
	\Else
	    \State $u \gets$ smallest distance between myself and a robot not on my position
	    \State $r \gets rbit()$
	    \State move by a distance of $\frac{u}{3}$ to the right if $r=0$, to the left otherwise
    \EndIf
\EndIf
\end{algorithmic}
\end{algorithm}

The pseudocode for our solution is presented as Algorithm~\ref{onescattering}. The term "unit distance" found in Line 6 refers to the one given by the scheduler. Figure~\ref{1D} presents the two possible locations a robot located in $A$ may chose as a target destination when occupied locations are only $A$, $B$, and $C$. 


\begin{figure}[h]
 \centering
 \includegraphics[width=0.9\textwidth]{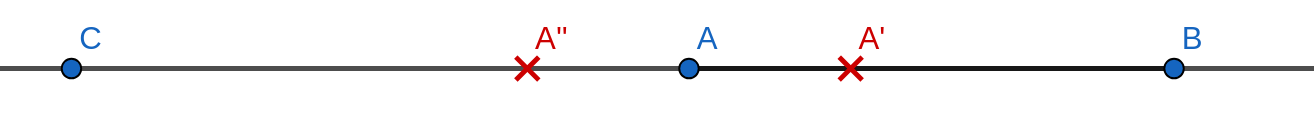}
 \caption{$A'$ and $A''$ are the two positions a robot located in $A$ can go to (assuming it's not alone on its position)}
 \label{1D}
\end{figure}

 
\begin{lemma}
$\sep$ is a safe zone for algorithm \ref{onescattering}.
\end{lemma}

\begin{proof}
Starting from a configuration in $\sep$, a robot does nothing when there is no other robot at its position (Line 2). 
\end{proof} 
 
\begin{theorem}
Algorithm \ref{onescattering} solves the ASYNC flexible scattering problem in $\R$.
\label{thm:RFlexScattering}
\end{theorem}
 
\begin{proof}
Our strategy is to increase by at least one the number of isolated (and hence immobile) robots at each step of the proof. 
We first fix a left-right orientation (unknown to the robots), and number the robots from $1$ to $n$. Let $\sigma$ be a scheduler and $C$ a configuration.
For the first step, we devise a strategy to put one robot to the right of every other robot (and it remains so forever): if a robot located at the rightmost position (called "rightmost robot") is requested to make a choice, then if no other rightmost robot already goes to the right, go to the right; otherwise, go to the left. We refer to this event as  ($\ast$). We now assume that event ($\ast$) always happens when it can, \emph{i.e.} our strategy occurs.
\begin{enumerate}
    \item This strategy occurs with probability at least $2^{-2n}$. Indeed, if a rightmost robot decides to go left once, it is never considered again as it cannot become a rightmost robot again.
    In more details, Algorithm~\ref{onescattering} requires that for a robot $A$ to cross a robot $B$ to the right, $B$ should go to the left. With our strategy, at least one rightmost robot goes right or stays still. 
    It follows that the number of times we need to repeat the strategy is at most $2n$, as at most $n$ rightmost robots may be requested to go to the left (and if a rightmost robot goes to the left it must be that a rightmost robot has chosen to go to the right).
    Then, we need each robot on the rightmost position stack to make the right choice between left and right, each one has a probability $\frac{1}{2}$ of making this right choice. So, the strategy is followed with probability at least $2^{-2n}$.
    \item There is a time after which no robot wants to go further to the right of or at the same position as the rightmost robot, and the rightmost robot does not want to move to the left. Said differently there is a time after which the rightmost robot remains the rightmost robot forever (we call this property $P$). This is entailed by two facts: \emph{(i)}, ($\ast$) can occur at most $2n$ times (see above) and if ($\ast$) occurs finitely many times, then property $P$ is satisfied. 
     It is easy to be convinced of this claim if one notices that a non-rightmost robot cannot become a rightmost robot via a computation (unless the rightmost robots all move left). Next, because we follow the strategy, the rightmost robots cannot go left (except maybe at the beginning).\\
     We now give the technical proof : 
    Assuming the robots follow our strategy, there is a time after which no one wants to go further to the right of or at the same position as the rightmost robot and the rightmost robot does not want to move to the left, as a result the rightmost robot will remain the rightmost robot forever (we call this property P).\\
    We will prove it by induction on the number of robots who may become the rightmost robot at least once.\\
    Firstly it is easy to be convinced there is a time $t_{right}$ after which the position of the rightmost robot can only increase.\\
    Secondly if at a time $t > t_{right}$ a robot is still (its position and its destination is the same) and is not one of the rightmost robots, it can never become one of the rightmost robot. We call this property S.\\
    Let us now prove property P : by contradiction, let us assume that "who the rightmost robot" is changes an infinite amount of time, let us place ourselves after ($\ast$) does not happen anymore. Since we assume ($\ast$) happens whenever it can happen then it means that there is no situation, after some time $t > t_{right}$, where one of the rightmost robot is asked to compute.\\
    We say that a robot is in the race for rightmost status if there is a scheduler such that following this scheduler and our strategy there is a non-zero probability that this robot becomes, at least for some time, the rightmost robot. In other words a robot is in the race for rightmost status if there is a way for this robot to become one of the rightmost robots at some point.\\
    We place ourselves at time t, we number from 1 to r the robots in the race for rightmost status. There is a time t' where the rightmost robot is crossed by the robot k (by assumption). There is a time t'' at which k stops (either because the scheduler interrupted him or because k reached its destination) at this point k's position and k's destination are the same. If k never moves again it is removed from the race. If k moves once again then it means there is a time $t_3$ where another robot q which goes on top of it and k performs a computation. If at time $t_3$ k is not the rightmost robot then it is removed from the race (by property S). But k is not the rightmost robot at time $t_3$ because $t_3 >t$, and by definition of t the rightmost robot cannot compute after t.
    Therefore k is removed from the race. By induction this proves claim P (P is equivalent to having only one robot in the race for rightmost status).
    \item After long enough the rightmost robot does not want to move : by the last point there is a time at which there the rightmost robot does not move or moves to the right by fairness of the scheduler it will eventually reach its destination or be interrupted and never be reached by another robot (by the previous point)
\end{enumerate}
The set of configurations $F_1$ were a single rightmost robot does not want to move, and every other robot does not want to cross it is safe (by points 2 and 3 above). Moreover from any starting configuration $C$ and for any scheduler $\sigma$, there is a probability at least $2^{-2n}$ to reach this set. By theorem of choice, $F_1$ is thus reached with probability $1$.

For the next step of the strategy, we assume that there exists some integer $k\geq 1$ such that the following property holds: "the set of configuration $F_k$ where the $k$ rightmost robots are lonely and do not want to move, and every other robot does not want to cross the leftmost of those $k$ rightmost robots is safe". If a configuration is in $F_k$, We call $R_k$ the set of robots minus the $k$ single rightmost robots. Then, we use a strategy similar to the base case of the induction: if you are a rightmost robot of $R_k$ and you are required to make a choice, if no other robots at your current location is already going to the right, then go to the right, otherwise, go to the left. By a similar set of arguments as for the base case, after long enough, a configuration in $F_{k+1}$ is reached, and this configuration is safe. 

Step by step, we conclude that we eventually reach a configuration in $F_n$, and this configuration is safe. Obviously, $F_n$ satisfies $\sep$.
\end{proof}

\clearpage

\subsection{Flexible ASYNC Scattering in $\R^2$}
\label{sec:bidimensional}


The pseudo-code for algorithm scattering is given as Algorithm~\ref{twoscattering}. The smallest enclosing polygon (line 10) is also known as the \emph{convex hull}. 
Cardinal directions (Line 7) are given by the scheduler, while the vector $\overrightarrow{e_x}$ (Line 14) denotes the first axis of the frame given by the scheduler to the robot. 

When we write "move alongside $\theta$", where $\theta$ is an angle, we mean that the robot should move in the direction given by the vector $\overrightarrow{v}$, of size one and where $angle( \overrightarrow{e_x}, \overrightarrow{v}) = \theta$. An illustration of the algorithm is given in Figure~\ref{2D}. 


\begin{algorithm}[htbp]
\caption{Flexible ASYNC Scattering in $\R^2$ for robot $i$}\label{twoscattering}
\begin{algorithmic}[1]
\If {there is no other robot on my position} 
	\State Do nothing
\Else
    \If{there is only one occupied location}
	    \State $r_1 \gets rbit()$, $r_2 \gets rbit()$
		\State move by a unit distance to the East if $r_1=0$ and $r_2=0$, to the North if $r_1=0$ and $r_2=1$, to the South if  $r_1=1$ and $r_2=1$, to the West otherwise.
	\Else
	    \State Let $A$ be my current position
	    \State Let $P$ be the smallest enclosing polygon of all robots
	    \State Let $G$ be the barycenter of $P$
	    \State For every robot (including myself) $r$, compute $H_r$, the homothety of $P$ that passes through $r$ with center $G$
	    \State $u \gets$ smallest distance between $H_i$ and another polygon or $G$
	    \State $\theta \gets angle( \overrightarrow{e_x}, \overrightarrow{GA}) $, $u \gets u/3$ 
	    \State $r_1 \gets rbit()$, $r_2 \gets rbit()$
	    \State Move by distance $u$, alongside $\theta + \frac{\pi}{10}$ if $r_1= 0$ and $r_2 = 0$, alongside $\theta - \frac{\pi}{10}$ if $r_1= 0$ and $r_2 = 1$, alongside $-\theta + \frac{\pi}{10}$ if $r_1= 1$ and $r_2 = 0$, alongside $-\theta - \frac{\pi}{10}$ if $r_1= 1$ and $r_2 = 1$.
    \EndIf
\EndIf
\end{algorithmic}
\end{algorithm}


\begin{figure}[htbp]
\centering
 \includegraphics[width=.85\textwidth]{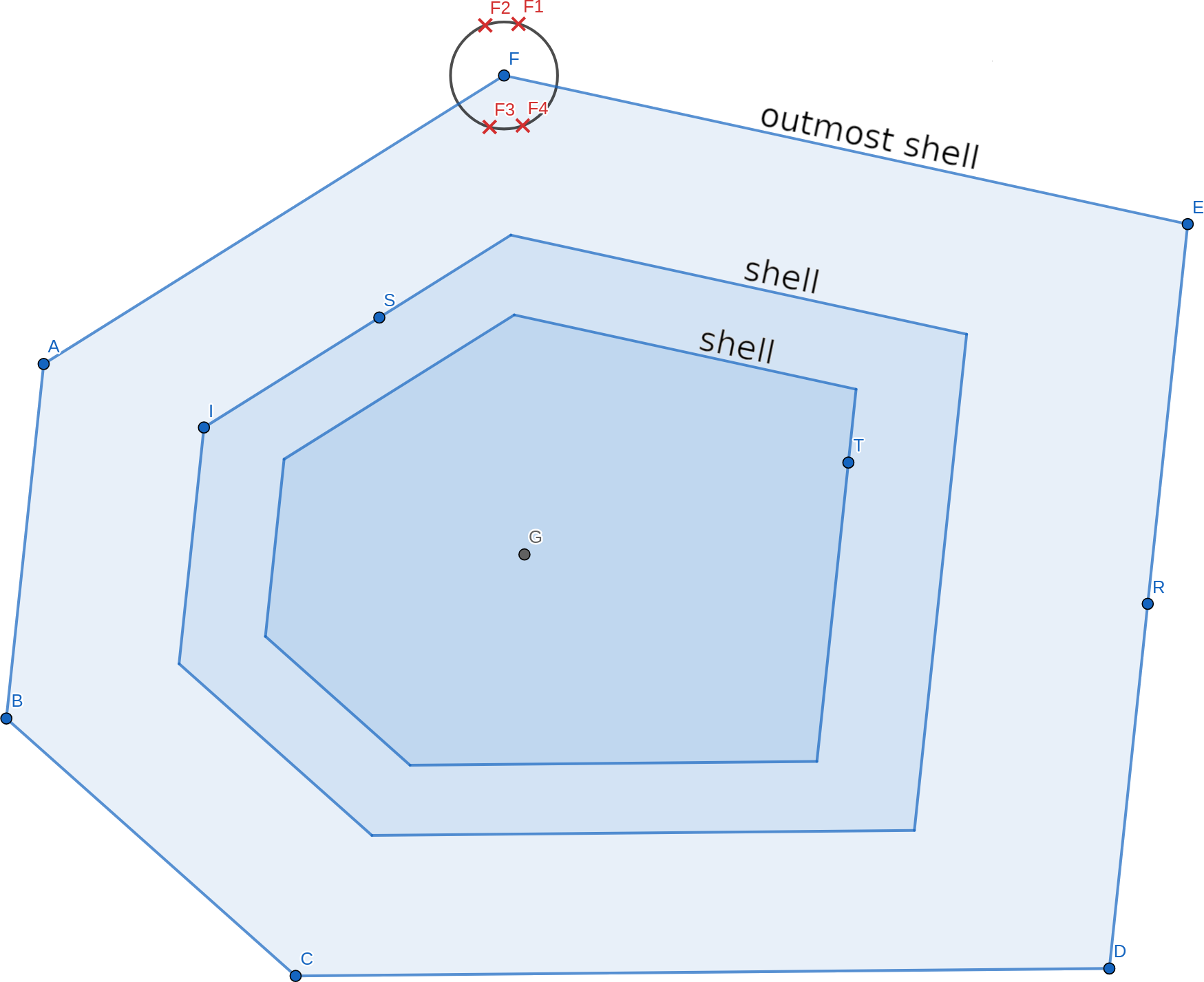}
 \caption{After observing the current configuration, robot $F$ (assumed not to be alone) computes the smallest enclosing polygon and its homotheties, and selects the four possible destinations $F1$, $F2$, $F3$, and $F4$.}
 \label{2D}
\end{figure}



\begin{lemma}
$\sep$ is a safe zone for Algorithm~\ref{twoscattering}.
\end{lemma}

\begin{proof}
Starting from a configuration in $\sep$, a robot does nothing when there is no other robot at its position (Line 2).
\end{proof}

\begin{theorem}
\label{thm:R2FlexScattering}
Algorithm \ref{twoscattering} solves the flexible ASYNC scattering problem in $\R^2$.
\end{theorem}

\begin{proof}
The proof for the $\R^2$ case follows the same pattern as the one for $\R$. However, it is now more ambiguous to talk about "farthest" robots. We use the smallest enclosing polygon to define the set of farthest robots, but instead of two directions (left and right) in the case of $\R$, we now have multiple angles to distinguish between farthest robots. In more details,let us denote by $\overrightarrow{e_x}$ the first coordinate vector of $\Rs$, let $C$ be a finite set of points in $\Rs$, let $\theta$ be a number in $\angles$, $c \in C$ is one of the $\theta$-furthest points if $ \langle c , e^{i \theta} \rangle = \underset{c'\in C}{\mathrm{max}} \langle c' , e^{i \theta} \rangle$ where $e^{i \theta}$ is the vector of length 1 and such that  $ \langle \overrightarrow{e_x}  , e^{i \theta} \rangle = \theta$ (see figure \ref{angle})

\begin{figure}
 \centering
 \includegraphics[width=.8\textwidth]{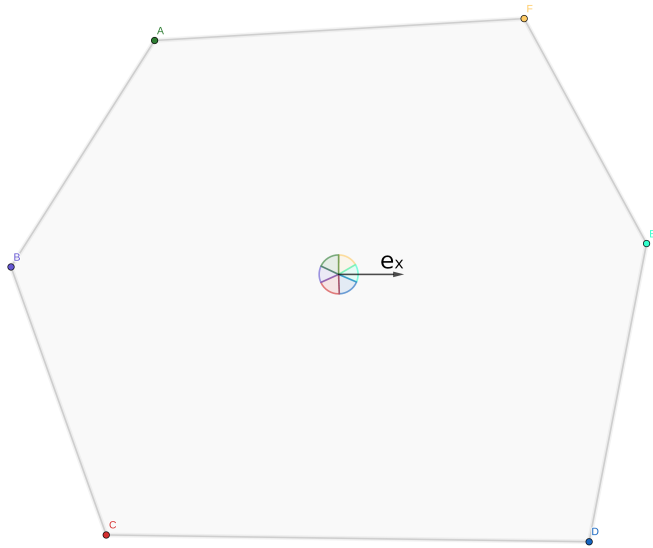}

 \caption{The displayed robots are the $\theta$-farthest robot for each $\theta$ that falls in a section of their color. $E$ is the $0^{\circ}$-farthest robot as well the $1^{\circ}$-farthest robot, $D$ is the $-45^{\circ}$-farthest robot, etc.}
 \label{angle}
\end{figure}

We number robots from 1 to $n$. Let $\sigma$ be a scheduler and $C$ a configuration. 

For the first step, we choose an angle $\theta$ arbitrarily, and we devise a strategy to put a robot the farthest along $\theta$ (called $\theta$-farthest). 
The strategy to follow is: if a robot $r$ is one of the $\theta$-furthest robots and is required to make a choice, if no other robot on the same location is already going $\theta$-further, $r$ pick a direction that makes it go $\theta$-further, otherwise $r$ picks a direction that does the opposite (that is, makes $r$ go $\theta$-closer). This event is called ($\ast$). Note that it is always possible to comply with the instructions because among the 4  directions robot $r$ can go to, at least one is making $r$ go $\theta$-further.
\begin{itemize}
    \item this strategy happens with a probability at least $\frac{1}{4}^{-2n}$. Indeed, if at some point $r$ decides to go $\theta$-closer, $r$ never applies the strategy again because $r$ cannot become a $\theta$-farthest robot again (by construction of the algorithm and the strategy). Therefore the event ($\ast$) happens at most $2n$ times, since there are 4 equally likely destinations and that at least one goes $\theta$-farther, we get $\frac{1}{4}^{-2n}$.
    \item After long enough, no robot wants to go $\theta$-further or equally $\theta$-far as the $\theta$-furthest robots (there can eventually be multiple $\theta$ farthest robots) : the proof is similar to the uni-dimensional case.
    Assuming the robots follow our strategy there is a time after which the set of $\theta$-farthest robots never change anymore. We call this property P.\\
    Firstly is it easy to be convinced that there is a time $t_{\theta}$ after which the value of the projection of the position of the $\theta$-farthest robots over the vector $e^{i\theta}$ only increases.\\
    Secondly if at a time $t>t_{\theta}$ a robot is still (its position and its destination is the same) and is not one of the $\theta$-farthest robots, it can never become one of the $\theta$-farthest robots. This is because, in the algorithm, the robot can move a distance which is strictly inferior to the distance to the closest enclosing polygon. We call this property S.\\
    Let us now prove property P : by contradiction, let us assume that the  set of $\theta$-farthest robots changes an infinite amount of time, let us place ourselves after ($\ast$) does not happen anymore. Since we assume ($\ast$) happens whenever it can happen then it means that there is no situation, after some time t > $t_{\theta}$, where one of the $\theta$-farthest robot is asked to compute.
    We say that a robot is in the race for $\theta$-farthest status if there is a scheduler such that following this scheduler and our strategy there is a non-zero probability that this robot becomes, at least for some time, the $\theta$-farthest robot.\\
    We place ourselves at time t, we number from 1 to r the robots in the race for $\theta$-farthest  status. Since the set of $\theta$-farthest robot always change by assumption, there is a time t' where robot k becomes a $\theta$-farthest robot (and it was not just before that). There is a time t'' at which k stops (either because the scheduler interrupted him or because k reached its destination) at this point k's position and k's destination are the same. 
    If k never moves again there are two possibilities 
    \begin{itemize}
        \item there is a point after which it is not a part of the set $\theta$-farthest robots and it is removed from the race (since it will never move) 
        \item or it will always be a part of the set of $\theta$-farthest robots.
    \end{itemize}
    
    If k moves once again then it means there is a time $t_3$ where another robot q which goes on top of it and k performs a computation.  If at time $t_3$ k is not the $\theta$-farthest robot then it is removed from the race (by property S). But k cannot be $\theta$-farthest robot at time $t_3$ because $t_3>t$ (and no robot $\theta$-farthest robot can compute after time t).
    Therefore robot k is removed from the race. By induction this proves claim P because each time we either discard a robot from ever being in the set of $\theta$-farthest robots or we put a robot in this set forever, at some point the set of $\theta$-farthest robot will never change anymore (which is property P).
    \item After long enough, the $\theta$-furthest robots do not want to move anymore. Once the situation described in the previous bullet case is reached, we just wait for the $ \theta$ furthest robots to reach their destination (or be stopped by the scheduler, for what matters).
\end{itemize}
 After events listed in the above point have happened, the set of $\theta$-farthest robots remains the same forever. \\*
 
Next steps of the strategy :

A polygon (or shell) is said stable if all robots on the border of the polygon are standing still and no one inside this polygon wants to go outside or on the border of the polygon.
We  successively apply the strategy we used for a specific $\theta$ to all $\theta ' \in \angles$. For a given $\theta$ it can be we are already at a point where the set of $\theta$-farthest robots will always stay the same and we have nothing to do. The cases where we do have something to do occur at most n times. Once the strategy has been applied for all $\theta ' \in \angles$ the outer shell (or convex hull) of the robots is stable. Just like in the uni-dimensional case we then apply the same strategy to the robots strictly inside the stable shell. We simply change in the reasoning "$\theta$-furthest"  by "$\theta$-furthest among the robots that are not part of a stable shell".\\

At the end, by application of the theorem of choice we reach $\sep$.


\end{proof}

Observe that Algorithm~\ref{twoscattering} very easily extends to higher dimension.



\clearpage

\section{Complexity}
\label{sec:complexity}

The notion of complexity can be studied in two variants: the expected number of coin tosses until scattering, and the expected number of rounds until scattering. A round is a smallest fragment of an execution during which all robots execute each at least one Look-Compute-Move cycle.

\begin{theorem}
Algorithm \ref{scattering} achieves scattering of $n$ robots in expected $O(2^{2n^2})$ rounds.
\end{theorem}

\begin{proof}
Let $E$ be the expected number of rounds it takes to scatter the robots. There is a $2^{-2n^2}$ probability that $\sep$ is reached in one round (see the proof of Lemma~\ref{sepalg}). Therefore, 
$E \leq 1 + (1-2^{-2n^2})*E \Rightarrow E \leq \mathcal{O}(2^{2n^2})$
\end{proof}

\begin{theorem}
Algorithm \ref{scattering} uses expected $\mathcal{O}(2n * 2^{n^2})$ coin tosses.
\end{theorem}

\begin{proof}
The proof is similar to the last theorem.
\end{proof}

By tweaking the algorithm \ref{sepalg}, we can get a better bound of $E \leq (\frac{4 \pi ^2}{6})^n * n^{2n}$. For this change, the probability of picking the $i$-th direction to $\sim \frac{6}{\pi^2 i^2}$ instead of $2^{-i}$.


We know that in the SSYNC setting, scattering can be achieved in $\Theta(n \log n)$ coin throws on average. 
In the ASYNC setting, assuming a robot throws a coin anytime it chooses a new destination, it turns out that $\Omega(n^2)$ coin tosses are needed for flexible scattering in $\R$.

\begin{theorem}
Assuming robots are not blind and have weak local multiplicity detection, any algorithm that solves flexible ASYNC scattering in $\R$ and makes a robot throw at least one coin anytime it chooses a new destination requires $\Omega(n^2)$ coin tosses.
\end{theorem}

\begin{proof}
\begin{figure}[h]
\centering
    \includegraphics[width=0.9\textwidth]{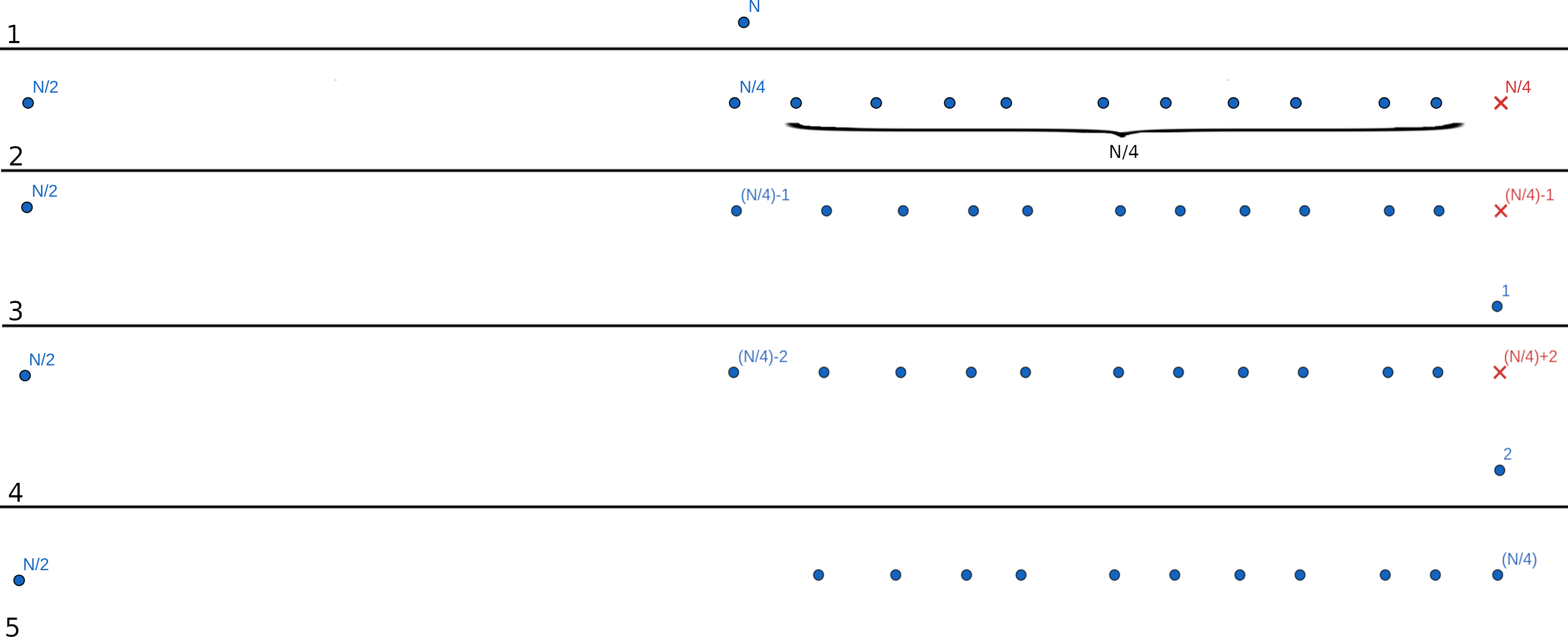}
    \caption{Steps of the proof}
    \label{1Dlower}
\end{figure}
At step $1$, suppose every robot is at the same position (at the origin of the coordinate system).
At step $2$, every robot computes. There is thus a point $p$ (denoted by a red cross on Figure~\ref{1Dlower}) on the left or on the right such that at least $n/2$ robots wish to cross $p$. Out of these $n/2$ robots, we disseminate $n/4$ along the path to $p$ by prematurely stopping them. The $n/4$ other robots are put on hold. At step $3$, we choose one of the $n/4$ robots that were put on hold at step $2$, and call it $q$. We now make $q$ move toward $p$ in an multi-stop fashion: anytime $q$ crosses one of the disseminated $n/4$ robots, say $q^\star$, $q$'s movement is suspended, and $q^\star$ computes a new destination (since $q^\star$ and $q$ occupy the same location at this moment). When $q$ reaches $p$, the scheduler stops it.

Each of the $n/4$ has made a computation transition (and hence a coin toss). The scheduler now moves and interrupts each of those in such a way that they remain between the origin and $p$. 
The process is then repeated with every remaining $n/4-1$ robots that were put on hold. As a result, each of the $n/4$ disseminated robots computes $n/4$ times, yielding a $\Omega(n^2)$ number of coin tosses. 
\end{proof}

Our assumption that a robot tosses a coin anytime it chooses a new destination is justified by the following observation: in our schedule, the scheduler only selects for computation robots that are on a multiplicity point; now if a robot on a multiplicity point does not perform a coin toss, its behavior may be identical to that of other robots on the same multiplicity point, delaying scattering.   

\clearpage

\section{Necessary and Sufficient conditions for ASYNC scattering}
\label{sec:computability}

In this section, we consider necessary and sufficient conditions for the solvability of ASYNC scattering depending on system hypotheses: rigidity of moves, dimensionality of the locations space, vision, and availability of weak local multiplicity detection. The summary of our results is presented in Table~\ref{resultstable}.

\begin{table}
\centering
{\small
\begin{tabular}{|l|l|c | c| c| c|}
\hline
\textbf{Vision} & \textbf{Weak Local} & \textbf{1D rigid}       & \textbf{2D rigid}       &  \textbf{1D flexible} & \textbf{2D flexible}      \\
& \textbf{Multiplicity}& & & & \\
\hline
Yes & Yes & OK (Th.\ref{thm:RFlexScattering}) & OK (Th.~\ref{theo2Drigid})  & OK (Th.~\ref{thm:RFlexScattering})  & OK (Th.~\ref{thm:R2FlexScattering})        \\
\hline
No & Yes & KO (Th.~\ref{no1Dscatblindrigid})     & OK (Th.~\ref{theo2Drigid})   & KO (Th.~\ref{no1Dscatblindrigid})   & ?        \\
\hline
Yes & No & KO (Th.~\ref{noWLno}) & ? & KO (Th.~\ref{noWLno}) & ? \\
\hline
No & No & KO (Th.~\ref{thm:impossible}) & KO (Th.~\ref{thm:impossible}) & KO (Th.~\ref{thm:impossible}) & KO (Th.~\ref{thm:impossible}) \\
\hline
\end{tabular}
}
\caption{Summary of possibility and impossibility results}
\label{resultstable}
\end{table}

\begin{theorem}
\label{no1Dscatblindrigid}
Rigid ASYNC scattering in $\R$ is impossible if the robots are blind, even with weak local multiplicity detection.
\end{theorem}

\begin{proof}
Consider three robots $A$, $B$, and $C$ that are still and $A$ is on the left of $B$, while $C$ shares $B$'s position. Let $d \in \R$ be a positive distance. Then, for any $p \in ]0,1[$, the scheduler can force a robot that computes a new destination to compute a distance at least $d$ with probability $p$ (indeed, the scheduler can change the scale of the coordinate system at every Look-Compute-Move cycle). 
We now construct a schedule such that there is a positive probability that the robots never reach $\sep$.
We define \emph{routine(i)} as "Assume there are exactly two robots $p$ and $q$ sharing the same location and at least one of them, $p$, is still, while the third robot is called $h$. We also assume that $q$ is not going toward $h$. If at least one of the assumptions is false, the routine aborts, otherwise, the scheduler does the following:
\begin{itemize}
\item If $q$ is still, the scheduler makes $p$ compute and ensures it chooses a destination that is farther than $dist(p,h)$ with probability $1-\frac{1}{2}^{(i+1)}$. If $p$ chooses to go towards $h$, the scheduler makes it move to the exact location of $h$, from which we execute \emph{routine(i+1)}. Otherwise, the scheduler simply executes \emph{routine(i+1)}.
\item If $q$ is not still, the scheduler makes $p$ compute and ensures it chooses a destination that is farther than $max(dist(h,p),dist(dest(q),p)$ with probability $1-\frac{1}{2}^{(i+1)}$. If $p$ goes towards $h$, the scheduler makes it move to the exact location of $h$ and executes \emph{routine(i+1)}. If $p$ goes farther than $q$, the scheduler makes $q$ go to its destination, then makes $p$ move to the location of $q$, and finally execute \emph{routine(i+1)}. If $q$ goes further than $p$, the scheduler  makes $p$ go to its destination, then makes $q$ move to the location of $p$ and execute \emph{routine(i+1)}."
\end{itemize}

Now, with the initial configuration we considered, let the scheduler execute \emph{routine(0)}. As we request the scheduler to ensure a certain moving distance with a certain probability, for any integer $i$, the \emph{routine(i)} never aborts, and the robots never reach $\sep$. The probability that this execution occurs is $\prod_{i=1}^{+ \infty} (1-\frac{1}{2}^{(i+1)}) \geq 1/2$. 
\end{proof}

\begin{theorem}
\label{noWLno}
Rigid ASYNC scattering in $\R$ is impossible if robots have no multiplicity detection, even with full vision. 
\end{theorem}

\begin{proof}
Let us first consider a network of three robots $A$, $B$, and $C$ such that $A$ is on the left, and $B$ and $C$ occupy the same location on the right. All robots are still. Now, when $B$ is activated, there must be a positive probability that $B$ computes a new position. Suppose $B$ only selects one direction: opposite to $A$. Then, the same holds for $C$, as the algorithm is uniform for all robots. The scheduler then uses the following schedule: \emph{(i)} activate $B$ until $B$ decides to move to the right by some distance $d_B>0$, then \emph{(ii)} activate $C$ until $C$ decides to move to the right by some distance $d_C>0$, then \emph{(iii.a)} if $d_B<d_C$, move $B$ until it reaches its destination, move $C$ until it reaches $B$, activate $B$ until it decides to go right again by some distance $d'_B>0$, then repeat the process until $d'_B>d_C-d_B$, or \emph{(iii.b)} if $d_B>d_C$, move $C$ until it reaches its destination, move $B$ until it reaches $C$, activate $C$ until it decides to go right again by some distance $d'_C>0$, then repeat the process until $d'_C>d_B-d_C$. This schedule yields an execution such that robots $B$ and $C$ occupy the same position infinitely often, hence scattering is never achieved.
This implies that robots $B$ and $C$ must run an algorithm such that there exists a positive probability to move to the left (that is, toward $A$). 

Now, consider a network of two robots $A$ and $B$ initially separated and still (hence satisfying $\sep$), where $A$ is on the right. Since $B$ has exactly the same vision as in the previous situation (it has no multiplicity detection), there is a positive probability that $B$ goes to the left (toward $A$). By symmetry, there is also a positive probability that $A$ goes to the right (that is, toward $B$).
Now, the distance can be oblivious of the distance between $A$ and $B$ (that is, an absolute value based on the unit distance given by the scheduler) or a function of the distance between $A$ and $B$ (since the scale is given by the scheduler, the perceived distance can be always one unit, so this equates to a fraction of the distance). In the case of absolute value, the scheduler simply chooses the distance between $A$ and $B$, so that they occupy the same position again, invalidating $\sep$ safety. In the case of a fraction, say $1/3$ (see Figure~\ref{1DnoWL}, the same number is chosen by both $A$ and $B$ with probability $p>0$ since they have the same view. Then $A$ computes a right target with probability $p$ but does not move, then $B$ computes and move four times in a row, hence, with probability $p^5$, $A$ and $B$ occupy the same position again, hence $\sep$ is not safe. If the fraction is $1/k$, for some integer $k$, the overall probability of a meeting of $A$ and $B$ becomes dependent of $k$, but remains positive.
\begin{figure}[h]
\centering
    \includegraphics[width=.9\textwidth]{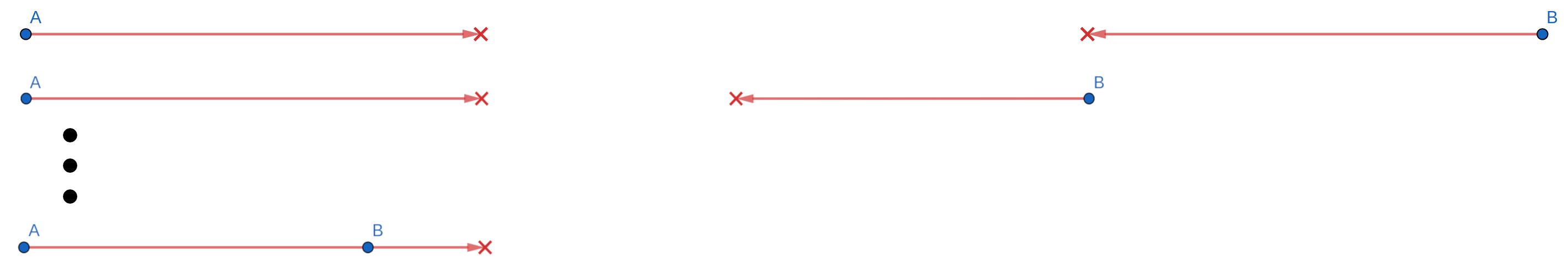}
    \caption{Case where $f=1/3$}
    \label{1DnoWL}
\end{figure}
\end{proof}

One may wonder if it is possible to scatter robots that are blind and have no weak local multiplicity detection (even in $\Rs$). The following theorem show that it is not the case.

\begin{theorem}
\label{thm:impossible}
Rigid ASYNC scattering in $\Rs$ is impossible if robots are blind and have no weak local multiplicity detection.
\end{theorem}

\begin{proof}
Since the robots do not have weak local multiplicity detection and are blind, there exists probability $p>0$ that the robots move, otherwise they would never separate if placed on the same location. 
Let $d\in \Rs$ be a position, since the robots are blind by adjusting the scale and orientation of their referential the scheduler can take the robot to position $d$ with probability $p>0$. Let us suppose we have $n>1$ robots in $\Rs$ that execute a scattering algorithm. Let us also suppose those robots are in a configuration that satisfies $\sep$. We now show that it is possible that two robots meet with positive probability. Let $d$ be the position of robot number $1$, the scheduler activates robot number $2$ such that robot $2$ has a probability $p$ of choosing $d$ as a destination. If robot $2$ chooses $d$ then we have reached a configuration not in $\sep$ with probability $p$. If robot $2$ does not choose $d$, the scheduler makes robot $2$ move to its destination and once arrived, activates it again so that it chooses $d$ with probability $p$. By doing this repeatedly, we eventually reach a conifguration not in $\sep$.

Admittedly, a fairness issue may arise, if robot $1$ is not activated fairly. However, the arguments remains the same if robot $1$ is activated and moves every $k$ (for some integer $k>0$) activations of robot $2$, as the probability $p$ is independent of the position of robot $1$.
\end{proof}

The first unresolved case in this table is whether it is possible to scatter blind robots in $\Rs$ with flexible moves. We conjecture the answer is yes and that Algorithm~\ref{scattering} is a solution. However, the proof argument is expected to be more involved than the one we provided. 
The second unresolved case is whether it is possible to scatter robots without weak local multiplicity detection in $\Rs$ (be it rigid or flexible moves). We also conjecture the answer to this question is yes, but different algorithmic techniques from the ones we provided are expected to be necessary.

\section{Concluding Remarks}
\label{sec:conclusion}

We presented the first solutions to the oblivious mobile robots asynchronous scattering problem. Contrary to previous work, we do not assume that the configuration observed by the robots is always up to date or does not features moving robots. It turns out that the problem can be solved in the most general case when robots are able to see other occupied location, and can determine whether more than one robot occupies their current location. Our positive results are constructive, and a byproduct of our approach is a new proof technique that could reveal useful for other oblivious mobile robot problems requiring randomization.
Perhaps surprisingly, it remains possible to solve the problem when robots are blind (they cannot see other occupied locations) in a bidimensional Euclidean space. Releasing both assumptions (vision and weak local multiplicity detection) yields immediate impossibility, while removing one of them yields impossibility in the case of line. Some cases remain open and most likely require new proof arguments.

While our approach was mostly driven by a computability perspective, our positive result command a thorough complexity analysis of the problem. Preliminary results we obtained show that, under a reasonable assumption, a lower bound of $\Omega(n^2)$ coin tosses is necessary to scatter $n$ robots in ASYNC. This contrasts with the weaker SSYNC model where $\Theta(n\log n)$ coin tosses are tight. We leave this path of research as a future work. 

\bibliographystyle{plain}
\bibliography{biblio}

\end{document}